\documentclass{llncs}
\usepackage{amsmath,amssymb,amsthm}
\usepackage{url}
\usepackage{epsfig}
\usepackage[latin1]{inputenc}
\usepackage[mathscr]{euscript}
\usepackage{subfigure}

\input xy \xyoption{all}

\begin{document}

\title{Differential Privacy versus Quantitative Information Flow\thanks{This work has been partially supported by the project ANR-09-BLAN-0169-01  PANDA  and by  the INRIA DRI  Equipe Associ\'ee PRINTEMPS.}
}
\author{ M\'ario S. Alvim$^1$ \and Konstantinos Chatzikokolakis$^2$ \and \\ Pierpaolo Degano$^3$ \and Catuscia~Palamidessi$^1$}
\institute{
$^1$ INRIA and LIX, Ecole Polytechnique, France.\\
$^2$ Technical University of Eindhoven, The Netherlands.\\
$^3$ Dipartimento di Informatica, Universit\`a di Pisa, Italy.
}

\maketitle

\begin{abstract}
Differential privacy is a notion of  privacy that has become very popular in the database community. Roughly, the idea is that a randomized query mechanism provides sufficient privacy protection if the ratio between the probabilities of two different entries to originate a certain answer is bound by $e^\epsilon$. In the fields of anonymity and information flow there is a similar concern for controlling information leakage, i.e. limiting the possibility of inferring the secret information from the observables. In recent years, researchers have proposed to quantify the leakage in terms of the information-theoretic notion of mutual information. There are two main approaches that fall in this category: One based on Shannon entropy, and one based on R\'enyi's min entropy. The latter has connection with the so-called Bayes risk, which expresses the probability of guessing the secret.

In this paper, we show how to model the query system in terms of an information-theoretic channel, and we compare the notion of differential privacy with that of mutual information. We show that the notion of differential privacy is strictly stronger, in the sense that it implies a bound on the mutual information, but not viceversa. 
\end{abstract}

\section{Introduction}
The growth of information technology raises significant concerns about the vulnerability of sensitive information. The possibility of collecting and storing data in large amount  and the availability of powerful data processing  techniques open the way to the threat of inferring private and secret information, to such an extent that fully justifies the users' worries. 

\subsection{Differential privacy}
The area of statistical databases has been, naturally, one of the first communities to consider the issues related to the protection of information. Already some decades ago, 
Dalenius  \cite{Dalenius:77:ST} proposed a famous ``ad omnia'' privacy desideratum: nothing about an individual should be learnable from the database that cannot be learned without access to the database.

Dalenius' property, however, is too strong to be useful in practice: it has been shown by Dwork  \cite{Dwork:06:ICALP} that no useful database can provide it.  
In replacement Dwork has proposed the notion of \emph{differential privacy}, which has had an extraordinary impact in the community. Intuitively, such notion is based on the idea that the presence or the absence of an item in the database should not  change in a significant way 
the probability of obtaining a certain answer for a given query \cite{Dwork:06:ICALP,Dwork:10:SODA,Dwork:10:CACM,Dwork:09:STOC}. 

In order to explain the concept more precisely, let us consider the typical scenario: we have databases whose entries are values (possibly tuples) taken from a given universe. 
A database can be queried by users which have honest purposes, but also by attackers trying to infer secret or private data.  
In order to control the leakage of secret information, the curator uses some randomized mechanism, which causes a certain lack of precision in the answers.
Clearly, there is a trade off between the need of obtaining answers as precise as possible for legitimate use, and the need to introduce some fuzziness to the purpose of confusing the attacker. 

Let $\mathcal{K}$ be the randomized function that provides the answers to the queries. We say that
$\mathcal{K}$ provides $\epsilon$-differential privacy if for all databases $D$ and $D'$, 
such that one is a subset of the other and the larger contains a single additional entry, 
and for all $S \subseteq \mathit{range}(\mathcal{K})$,  the ratio between the probability that the result of $\mathcal{K}(D)$ is in $S$, and 
the probability that the result of $\mathcal{K}(D')$ is in $S$, is at most $e^\epsilon$.

Dwork has also studied sufficient conditions for a randomized function $\mathcal{K}$ to implement a mechanism satisfying $\epsilon$-differential privacy.
It suffices to consider a Laplacian distribution with variance depending on $\epsilon$, and mean equal to the correct answer \cite{Dwork:10:CACM}. 
This is a technique quite diffused in practice. 

\subsection{Quantitative information flow and anonymity}
The problem of preventing the leakage of secret information has been a pressing concern 
also in the area of software systems, and has motivated a very active line of research  called \emph{secure information flow}. 
Similarly to the case of privacy, also in this field, at the beginning, the goal was ambitious:  to ensure \emph{non-interference}, 
which means complete lack of leakage. But, as for Dalenius' notion of privacy, no-intereference is too strong for being obtainable 
in practice, and the community has started exploring weaker notions. Some of the most popular approaches are the quantitative ones, 
based on information theory. See for instance \cite{Clark:01:QAPL,Clark:05:JLC,Clarkson:09:JCS,Kopf:07:CCS,Malacaria:07:POPL,Malacaria:08:PLAS,Smith:09:FOSSACS}. 

Independently the field of anonymity, which is concerned with the protection of the identity of agents performing certain tasks,  has evolved towards similar approaches.  
In the case of anonymity it is even more important to consider a quantitative  formulation, because \emph{anonymity protocols typically use randomization} to obfuscate the link 
between the \emph{culprit} (i.e. the agent which performs the task) and the observable effects of the task. 
The first notion of anonymity, due to Chaum \cite{Chaum:88:JC}, required that the observation would not change the probability of an individual to be the culprit. 
In other words,  the protocol should guarrantee that the observation does not increase the chances of learning the identity of the culprit. 
This is very similar to Dalenius' notion of privacy, and equally unattainable in practice (at least, in the majority of real situations).  
Also in this case, researchers in the area have started considering weaker notions based on information theory, see for instance \cite{Chatzikokolakis:08:IC,Moskowitz:03:WPES,Zhu:05:ICDCS}.

If we abstract from the kind of secrets and observables, anonymity and of information flow are similar problems: there is some information that we want to keep secret,  
there is a system that produces some kind of observable information  depending on the secret one, and we want to prevent as much as possible that an attacker may 
infer the secrets from the observables. 
It is therefore not surprising that the foundations of the two fields have converged towards the same information theoretical approaches. 
The majority of these approaches are based on the  idea of representing the system (or protocol) as an information-theoretic channel taking the secrets in input ($X$) and producing the 
observables in output ($Y$). The \emph{entropy} of $X$, $H(X)$, represents the converse of the \emph{a priori vulnerability}, i.e. the chance of the attacker to find out the secret. Similarly, the conditional entropy of $X$ given $Y$, $H(X\mid Y)$, represents the converse of the \emph{a posteriori vulnerability}, i.e. the chance of the attacker to find out the secret after having observed the output. 
The  \emph{mutual information} between $X$ and $Y$, $I(X;Y) = H(X)-H(X\mid Y)$,  represents the gain for the adversary provided by the observation, and is taken as definition 
of the \emph{information leakage} of the system. Sometimes we may want to abstract from the distribution of $X$, in which case we can use the \emph{capacity} 
of the channel, defined as the maximum of  $I(X;Y)$ over all possible  distributions on $X$. This represents the worst case for leakage.

The various approaches in literature differ, mainly, for the notion of entropy. Such notion is related to the kind of attackers we want to model, and to how we measure their success (see \cite{Kopf:07:CCS} for an illuminating discussion of such relation).  Shannon entropy~\cite{Shannon:48:Bell}, on which most of the approaches are based,  represents an adversary which tries to find out the secret $x$ by asking questions of the form ``does $x$ belong to set $S$?''. Shannon entropy is precisely the average number of questions necessary to find out the exact value of $x$ with an optimal strategy (i.e. an optimal choice of the $S$'s). 
The other most popular notion of entropy (in this area)  is R\'enyi's min entropy \cite{Renyi:61:Berkeley}. The corresponding notion of attack is a \emph{single try} of the form ``is $x$ equal to $v$?''. 
R\'enyi's min entropy is precisely the log of the probability of guessing the true value with the optimal strategy, which consists, of course, in selecting the $v$ with the highest probability. 
Approaches based on this notion include \cite{Smith:09:FOSSACS} and \cite{Braun:09:MFPS}. 

It is worth noting that, while  the R\'enyi's min entropy of $X$, $H_\infty(X)$, represents the a priori probability of success (of the single-try attack), the R\'enyi's min conditional entropy of $X$ given $Y$, $H_\infty(X\mid Y)$, represents the a posteriori probability of success\footnote{We should mention that R\'enyi did not define the conditional version of the min entropy, and that there have been various different proposals in literature for this notion. We use here the one proposed by Smith in \cite{Smith:09:FOSSACS}.}. This a posteriori probability  is the converse of the Bayes risk~\cite{Cover:06:BOOK} , which has also been used as a measure of leakage \cite{Braun:08:FOSSACS,Chatzikokolakis:08:JCS}. 

\subsection{Goal of the paper}
From a mathematical point of view, privacy presents many similarities with information flow and anonymity. The private data of the entry constitute the secret, the answer to the query gives the observation, and the goal is to prevent as much as possible the inference of the secret from the observable. Differential privacy can be seen as a quantitative definition of the degree of leakage. 
The main goal of this paper is to explore the relation with the alternative definitions based on information theory, with the purpose of getting a better understanding of the notion of differential privacy, of the specific problems related to privacy, and of the models of attack used to formalize the notion of privacy, in relation to those used for anonymity and information flow. 

\subsection{Contribution}
The contribution of this paper is as follows:
\begin{itemize}
\item We show how the problem of privacy can be formulated in an information-theoretic setting. More precisely, we show how the  answer function $\mathcal{K}$ can be associated to an information-theoretic channel. 

\item We prove that  $\epsilon$-differential privacy implies a bound on the Shannon mutual information of the channel, and that this bound approach $0$ as $\epsilon$ approaches $0$. 
Same for R\'enyi min mutual information. 

\item We show that the viceversa of the above point does not hold, i.e. that Shannon and R\'enyi min mutual information  (and also the corresponding capacities)  can approach $0$ while the $\epsilon$ parameter of differential privacy approaches infinity. 


\end{itemize}

\subsection{Plan of the paper}
Next section introduces some necessary background notions. Section 3 proposes an information-theoretic view of the database query systems. 
Section 4 show the main results of the paper, namely that differential privacy implies a bound on Shannon and R\'enyi min mutual information, but not viceversa. 
Section 5 concludes and presents some ideas for future work.

The proofs of the results are in the appendix. Such appendix will not be included in the proceeding version (for reasons of space), but the proofs will be made available on line.
\section{Preliminaries}

\subsection{Differential privacy}

We assume a fixed finite universe $U$ in which the entries of databases may range.
The concept of differential privacy is tightly connected to the concept of \emph{adjacent} (or \emph{neighbor}) databases. 

\begin{definition}[\cite{Dwork:10:CACM}]
	\label{def:adjacent-db}
	A pair of databases $(D', D'')$ is considered \emph{adjacent} (or \emph{neighbors}) if one is a proper subset of the other and the larger database contains just one additional entry.
\end{definition}  

Dwork's definition of differential privacy is the following:

\begin{definition}[\cite{Dwork:06:ICALP}]
	\label{def:diff-privacy-1}
	A randomized function $\mathcal{K}$ satisfies \emph{$\epsilon$-differential privacy} if for all pairs of adjacent databases $D'$ and $D''$, and all $S \subseteq Range(\mathcal{K})$,
	\begin{equation}
		Pr[\mathcal{K}(D') \in S] \leq e^{\epsilon} \times Pr[\mathcal{K}(D'') \in S]		
	\end{equation}	
\end{definition}

\subsection{Information theory and interpretation in terms of attacks}

In the following, $X, Y$ denote two discrete random variables with carriers ${\cal X} = \{x_{1}{}, \ldots, x_{n}{}\},
\ 	{\cal Y} = \{ y_{1}{}, \ldots, y_{m}{} \}$, and
 probability distributions $p_{X}(\cdot)$,  $p_{Y}(\cdot)$,
 respectively. 
 An information-theoretic channel is constituted by an input $X$, an output $Y$, and the matrix of conditional probabilities
 $p_{Y\mid X}(\cdot \mid \cdot)$, where $p_{Y\mid X}(y \mid x)$ represent the probability that $Y$ is $y$ given that $X$ is $x$. 
We will use  $X\wedge Y$ to represent the random variable with carrier ${\cal X}\times{\cal Y}$ and joint probability distribution
 $p_{X \wedge Y}(x,y) = p_{X}(x) \cdot p_{Y\mid X}(y\mid x)$.
  We shall omit the subscripts on the probabilities when they are clear from the context.

 \subsection{Shannon entropy}
	
	The Shannon entropy of  $X$  is defined as
	\[
	H(X) = -\sum_{x\,  \in \, {\cal X}} p(x)
	\log \,p(x)
	\] 
	The minimum
	value $H(X) = 0$ is obtained when $p(\cdot)$ is concentrated on a single value (i.e. when $p(\cdot)$ is a delta of Dirac).
	The maximum value $H(X) = \log {|{\cal X}|}$ is obtained
	when $p(\cdot)$ is the
	uniform distribution. Usually the base of the  logarithm is set to  be $2$ and, correspondingly,
	the entropy is measured in \emph{bits}.
	
	The \emph{conditional entropy} of $X$ given $Y$
	 is 
	 \[
	 H(X\mid Y)  = 
	 	{\displaystyle\sum_{y\,\in\, {\cal Y}}
		p(y) \  H(X \mid Y = y )}
	\]
	where
	\[
	 H(X \mid Y = y )  = {\displaystyle -\sum_{x\,\in\,{\cal X}}
		p(x\mid y)
		\log \, p(x\mid y)  }
	\]
	We can prove
	that $0 \leq H(X \mid Y) \leq H(X)$. The minimum value, $0$,  is obtained when $X$ is completely
	determined by $Y$. The maximum value $H(X)$ is obtained when
	$Y$ reveals no information about $X$, i.e. when $X$ and $Y$ are independent.
	
	The \emph{mutual information} between $X$ and $Y$ is defined
	as 
	\begin{equation}\label{eqn:ShannonMutualInfo}
	I(X;Y) \ = \ H(X) - H(X \mid Y)
	\end{equation}
	and it measures
	the amount of  information
	about $X$ that we gain by observing $Y$. It can be shown that
	$I(X;Y) = I(Y;X)$ and $0 \leq I(X;Y) \leq H(X)$.
	
	Shannon capacity is defined as the maximum mutual information over all possible input distributions:
	\[
	C = \max_{p_{X}(\cdot)}I(X;Y)
	\]

 \subsection{R\'enyi min-entropy}\label{sec:minEntropy}
	
	In \cite{Renyi:61:Berkeley}, R\'enyi introduced an one-parameter family of entropy measures, 
	intended as a generalization of Shannon entropy.
	The R\'enyi entropy of order $\alpha$ ($\alpha > 0$,  $\alpha \neq 1$) of  a random variable $X$ is defined as
	\[
	H_\alpha(X) \ =\  \frac{1}{1-\alpha}\log\sum_{x \,\in\,{\cal X}} p(x)^\alpha
	\]
	We are particularly interested in the limit of $H_\alpha$ as $\alpha$ approaches $\infty$. This is  called \emph{min-entropy}. It can be proven  that
	\[
	H_\infty(X) \ \stackrel{\rm def}{=}\ \lim_{\alpha\rightarrow \infty}H_\alpha(X) \ =\  - \log\,\max_{ x\in{\cal X}}\,p(x)	
	\]
	
	R\'enyi defined also  the $\alpha$-generalization of other information-theoretic notions, like 
	 the Kullback-Leibler divergence. However, he did not define the $\alpha$-generalization of the  conditional entropy, and there is no agreement on what it should be. 
	For the case $\alpha = \infty$, we adopt here the definition of conditional entropy proposed by Smith in \cite{Smith:09:FOSSACS}:
	\begin{equation}\label{eqn:SmithCondEntropyInfty}
	 \begin{array}{lcl}
	H_\infty(X\mid Y) \ = \  - \log \sum_{y\in {\cal Y}} p(y)\max_{x\in {\cal X}} \ p(x \mid y)
	\end{array}
	\end{equation}
Analogously  to (\ref{eqn:ShannonMutualInfo}), we can define the mutual information 
 $I _\infty$ as $H_\infty(X) - H_\infty(X\mid Y)$, and the capacity $C_\infty$ as  $\max_{p_{X}(\cdot)}I_\infty(X;Y)$. 
It has been proven in~\cite{Braun:09:MFPS} that $C_\infty$ is obtained at the uniform distribution, and that it is equal to the sum of the maxima of each column in the channel matrix:
\[
C_\infty = \sum_{y \,\in\,{\cal Y}} \max_{x \,\in\,{\cal X}} p(y\mid x).
\]	

%

\section{An information theoretic model of privacy}
In this section we show how to represent a database query system (of the kind considered in differential privacy) in terms of an information-theoretic channel.


According to~\cite{Dwork:06:ICALP} and~\cite{Dwork:10:CACM}, differential privacy can be implemented by adding some  appropriately chosen random noise to the answer  $x=f(D)$, 
where $f$ is the \emph{query function} and $D$ is the database. The function can operate in the entire database at once, and even though the query may be composed by a chain of sub-queries, we assume that subsequent sub-queries depend only on the \emph{true answer} to previous sub-queries. Under this constraint, no matter how complex the query is, it is still a function $f$ of the database $D$. The scenario where subsequent sub-queries can depend on the \emph{reported answer} to previous queries corresponds to adaptive adversaries~\cite{Dwork:06:ICALP}, and is not considered in this paper.
	
	After the true answer $x$ to the query is obtained from $D$, some noise is introduced in order to produce a reported answer $y$. The reported answer can be seen as a random variable $Y$ dependent on the random variable $X$ corresponding to the real answer, and the two random variables are related by a conditional probability distribution $p_{Y|X}(\cdot|\cdot)$. The conditional probabilities
	$p_{Y|X}(y|x)$ constitute the matrix of an information theoretic channel from $X$ to $Y$.
	
	Figure~\ref{fig:query-flow-complete} shows the scheme of implementation of a differential privacy scheme.
 
	\begin{figure}[htb]%
		\centering
		\includegraphics[]{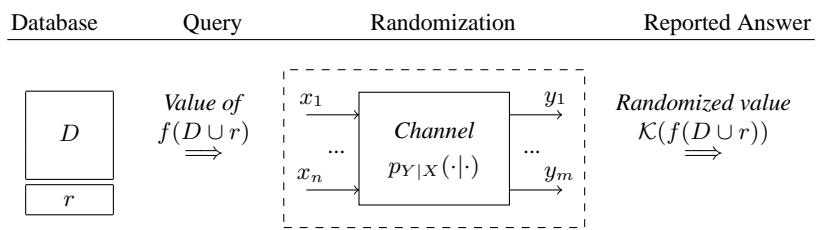}%
		\caption{The channel corresponding to a differential privacy scheme.}%
		\label{fig:query-flow-complete}%
	\end{figure}
	
	In \cite{Dwork:06:ICALP} it has been proved that a way to define the values of $p_{Y|X}(\cdot|\cdot)$ so to ensure $\epsilon$-differential privacy, is by using the Laplace distribution:
	
	\begin{equation}
		P((Y=y)|(X=x),\Delta f/\epsilon) = \frac{\Delta f}{2\epsilon} e^{-| y-x | \epsilon/ \Delta f}
	\end{equation}
	
	where $\Delta f$ is the L1-sensitivity of $f$, defined as\footnote{We give here the definition for the case in which the range of $f$ is $\mathbb{R}$. In the more general case in which the range is $\mathbb{R}^n$ we should replace   $| f(D') - f(D'') |$ by the $1$-norm of the vector $f(D') - f(D'')$.}
	\[\Delta f = \max_{D',D'' \,\mathit{adjacent}} | f(D') - f(D'') |.\]
	

\section{Relation between differential privacy and mutual information}

In this section we investigate the relation between differential privacy and information-theoretic notions.  We start by considering  an equivalent  definition of differential privacy, easier to handle for our purposes.

\subsection{Testing single elements}

Definition~\ref{def:diff-privacy-1} considers  tests which check whether the result of $\mathcal{K}(D)$ belongs to a certain set or not. 
We prefer to simplify this definition by considering only tests over single elements:

\begin{definition}
	\label{def:diff-privacy-2}
	A randomized function $\mathcal{K}$ gives \emph{$\delta$-differential privacy} if for all pairs adjacent datasets $D'$ and $D''$, and all $k \in Range(\mathcal{K})$,
	\begin{equation}
		Pr[\mathcal{K}(D') = k] \leq e^{\delta} \times Pr[\mathcal{K}(D'') =k]
	\end{equation}	
\end{definition}

The following result shows that our definition of differential privacy is equivalent to the classical one.
\begin{theorem}\label{theo:equivalent}
	A function $\mathcal{K}$ gives $\epsilon$-differential privacy iff it gives $\delta$-differential privacy, with $\epsilon = \delta$.
\end{theorem}

\subsection{Databases with the same number of entries and differing in at most one entry}

Consider two databases $D'$ and $D''$ that have the same number of entries and differ in at most one entry as in Figure~\ref{fig:databases}. Let $D$ be the common part shared by both databases, and let  $r'$ and $r''$ be the rows in which they differ, namely $D' = D \cup \{r'\}$ and $D'' = D \cup \{r''\}$.
	
	\begin{figure}[!htb]%
		\centering
		\includegraphics{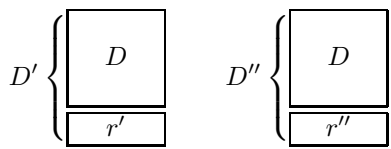}%
		\caption{Two databases differing in exactly one entry}%
		\label{fig:databases}%
	\end{figure}
	
We prove that  $\delta$-differential privacy imposes also 
a bound on the comparison between databases with the same number of entries, and which differ in the values of only one entry. 
\begin{lemma}\label{differ}
	Let $\mathcal{K}$ be a function that gives $\delta$-differential privacy for all pairs of adjacent databases. Given two databases $D'$ and $D''$ that have the same number of entries and differ in the value of at most one entry, then:
	$$ Pr[\mathcal{K}(D') = k] \leq e^{2\delta} \times Pr[\mathcal{K}(D'') = k] $$
\end{lemma}

\subsection{Shannon mutual information}\label{sec:Shannon}
We prove now that $\delta$-differential privacy imposes a bound on Shannon  mutual information, and that this bound approaches $0$  as the parameter $\delta$ approaches 
$0$.
\begin{theorem}\label{theo:Shannon}
	If a randomized function $\mathcal{K}$ gives $\delta$-differential privacy according to Definition~\ref{def:diff-privacy-2}, then for every result $x^*$ of the function $f$ the Shannon mutual information between the true answers $X$ (i.e. the results of $f$) and the reported answers $Y$  (i.e. the results of $\mathcal{K}$) is bounded by:
	$$I(X;Y) \leq (e^{2\delta} + e^{-2\delta})\delta \log(e) + (e^{2\delta} - e^{-2\delta})\sum_{y}p(y|x^*) \log(p(y|x^*))$$
\end{theorem}

It is easy to see that the expression which bounds $I$ from above,  $(e^{2\delta} + e^{-2\delta})\delta \log(e) + (e^{2\delta} - e^{-2\delta})\sum_{y}p(y|x^*) \log(p(y|x^*))$,  converges to $0$ when $\delta$ approaches $0$.

The converse of Theorem~\ref{theo:Shannon} does not hold. One reason is that mutual information is sensitive to the values of the input distribution, while differential privacy is not. Next example illustrates this point. 

\begin{example}\label{exa:sensitiveness}
Let $n$ be the number of elements of the universe, and $m$ the cardinality of the set of possible answers of $f$. 
Assume that $p(x_1) = {\alpha}$ and $p(x_i) = \frac{1-\alpha}{n-1}$ for $2\leq i\leq n$. Let $p(y_1\mid x_1) = \beta$, $p(y_j\mid x_1) = \frac{1-\beta}{m-1}$ for $2\leq j \leq m$, and  $p(y_j\mid x_i) = \frac{1}{m}$. otherwise. This channel is represented in Figure~\ref{tab:sub1}. It is easy to see that the Shannon mutual information approaches $0$ as $\alpha$ approaches $0$, independently of the value of $\beta$. 
Differential privacy, however, depends only on the value of $\beta$, more precisely, the parameter of differential privacy is $\max\{\log_e\frac{1}{m\beta},\log_e{m\beta},\log_e\frac{m-1}{m(1-\beta)},\log_e\frac{m(1-\beta)}{m-1}\}$, and it is easy to see that such parameter is unbound and goes to infinity as $\beta$ approaches $0$.
\end{example}

The reasoning in the counterexample above is not valid anymore if we consider capacity instead than mutual information. However, there is another reason why the converse of Theorem~\ref{theo:Shannon} does not hold, and this remains the case also if we consider capacity. The situation is illustrated by the following example.  

\begin{example}\label{exa:beta}
Let $n$ be the number of elements of the universe, and $m$ the cardinality of the set of possible answers of $f$. 
Assume that $p(y_i\mid x_i) = {\beta}$ and $p(y_i\mid x_j) = \frac{1-\beta}{m-1}$ for $i \neq j$.  This channel is represented in Figure~\ref{tab:sub2}.
It is easy to see that the Shannon capacity is $C= \log m - (1-\beta)\log (m-1) +\beta\log\beta+(1-\beta)\log(1-\beta)$, and that $C$ approaches $0$ as $\beta$ approaches $0$ and $m$ becomes large. Differential privacy, however, goes in the other direction when $\beta$ approaches $0$, and it is not very sensitive to the value of $m$. More precisely, the parameter of differential privacy is $\max\{\log_e\frac{1-\beta}{\beta(1-m)},\frac{\beta(1-m)}{1-\beta}\}$, and it is easy to see that such parameter is unbound and goes to infinity as $\beta$ approaches $0$, independently of the value of $m$.
\end{example}

\begin{table}[!htb]
		\centering
		\subtable[Example~\ref{exa:sensitiveness}]{
			 $\begin{array}{r}
				p_X(\cdot) \\ 
				 \\ [-2.5ex]
				\alpha \\ 
				 \\ [-2.5ex]
				 \\ [-2.5ex]
				\frac{1-\alpha}{m-1} \\ 
				 \\ [-2.5ex]
				 \\ [-2.5ex]
				\vdots \\ 
				 \\ [-2.5ex]
				 \\ [-2.5ex]
				\frac{1-\alpha}{m-1} \\ 
			\end{array}$
			$\begin{array}{|c||c|c|c|c|}
				\hline
				 & y_1 & y_2 & \ldots & y_m \\ \hline
				\hline
				& & & & \\ [-2.5ex]
				x_1 & \beta & \frac{1-\beta}{m-1} & \ldots & \frac{1-\beta}{m-1} \\ 
				& & & & \\ [-2.5ex]
				\hline						
				& & & & \\ [-2.5ex]
				x_2 & \frac{1}{m} & \frac{1}{m} & \ldots & \frac{1}{m} \\ 
				& & & & \\ [-2.5ex]
				\hline
				& & & & \\ [-2.5ex]
				\vdots & \vdots & \vdots & \ddots & \vdots \\ 
				& & & & \\ [-2.5ex]
				\hline
				& & & & \\ [-2.5ex]
				x_n & \frac{1}{m} & \frac{1}{m} & \ldots & \frac{1}{m} \\ 
				\hline
			\end{array}$
			\label{tab:sub1}
		} \quad \quad \quad
		\subtable[Example~\ref{exa:beta}]{
			$\begin{array}{|c||c|c|c|c|}
			 	\hline
			 	& y_1 & y_2 & \ldots & y_m \\ \hline
				\hline
				& & & & \\ [-2.5ex]
				x_1 & \beta & \frac{1-\beta}{m-1} & \ldots & \frac{1-\beta}{m-1} \\ 
				& & & & \\ [-2.5ex]
				\hline						
				& & & & \\ [-2.5ex]
				x_2 & \frac{1-\beta}{m-1} & \beta & \ldots & \frac{1-\beta}{m-1} \\ 
				& & & & \\ [-2.5ex]
				\hline						
				& & & & \\ [-2.5ex]
				\vdots & \vdots & \vdots & \ddots & \vdots \\ 
				& & & & \\ [-2.5ex]
				\hline
				& & & & \\ [-2.5ex]
				x_n & \frac{1-\beta}{m-1} & \frac{1-\beta}{m-1} & \ldots & \beta \\ 
				\hline						
			\end{array}$
			\label{tab:sub2}
		}
		\label{tab:examples}
		\caption{The channels of Examples~\ref{exa:sensitiveness} and~\ref{exa:beta}}
	\end{table}

\subsection{R\'enyi min mutual information}
We show now that a result analogous to that of Section~\ref{sec:Shannon} holds also in the case of R\'enyi min entropy.
\begin{theorem}\label{theo:Renyi}
	If a randomized function $\mathcal{K}$ gives $\delta$-differential privacy according to Definition~\ref{def:diff-privacy-2}, then the R\'enyi min mutual information between the true answer of the function $X$ and the reported answer $Y$ is bounded by
	$$I_{\infty}(X;Y) \leq 2\delta \log e.$$
\end{theorem}

The converse of Theorem~\ref{theo:Renyi} does not hold, not even if we consider capacity instead than mutual information. 
It is easy to prove, in fact, that Examples~\ref{exa:sensitiveness} and \ref{exa:beta} lead to counterexamples also in the case of R\'enyi min mutual information and capacity. 

\section{Conclusion and future work}

In this paper we have shown that the problem of privacy in statistical databases 
can be formulated in information-theoretic terms, in a way analogous to what has 
been done for information flow and anonymity: the database query system can be seen as a noisy channel, in the information-theoretic sense. 
Then we have considered Dwork's notion of differential privacy, and we have shown that it is strictly stronger than requiring the channel to have low capacity, 
both for the cases of Shannon and R\'enyi min entropy. 
It is natural to consider, then, whether a weaker notion would give enough privacy guarrantees. 
As future work, we intend to investigate this question. 

We first need to understand, of course, what are the constraints that could be relaxed in the notion of differential privacy. 
To this aim, Example~\ref{exa:beta} is quite interesting:  whenever we get an answer $y$, there are $n-1$ possible inputs (entries) which are equally likely to have 
generated that answer, and one input $x$ that is much less likely than the others ($p(x|y) = \alpha$, where $\alpha$ is a very small value). 
The existence of the latter seems quite harmless, yet it is exactly that entry that causes  differential privacy to fail (in the sense that its parameter is unbound). 
The notion of R\'enyi min capacity seems a plausible candidate for the notion of privacy: it's relation with the Bayes risk ensures that a bound $C_\infty$ can be seen as a bound on the probability of  guessing 
the right value of $x$ (given the obsevable). In some scenario, this may be exactly what we want. 

\subsection*{Acknowledgement}
We wish to thank Daniel Le M\'etayer for having pointed out to us the notion of differential privacy, and brought to our attention  the possible relation with quantitative information flow. 

\bibliographystyle{plain}
\bibliography{short}

\newpage
\section*{Appendix}

\begin{theorem} [Theorem~\ref{theo:equivalent} in the paper]
	A function $\mathcal{K}$ gives $\epsilon$-differential privacy iff it gives $\delta$-differential privacy, with $\epsilon = \delta$.
\end{theorem}

\begin{proof}
	
	\ \\
	
	$\Rightarrow$ Let $k \in Range(\mathcal{K})$. Then for all pair of adjacent databases $D', D''$ we have
	
	$$
	\begin{array}{lclll}
		\displaystyle Pr[\mathcal{K}(D') = k] & =    & \displaystyle Pr[\mathcal{K}(D') \in \{k\}] & \quad& \mbox{(taking $S$ to be a singleton set)} \\
		                                      & \leq & \displaystyle {e^{\epsilon}} Pr[\mathcal{K}(D'') \in \{k\}] && \mbox{(by Definition~\ref{def:diff-privacy-1})} \\
		                                      & =    & \displaystyle {e^{\epsilon}} Pr[\mathcal{K}(D'') = k] && \mbox{} \\
	\end{array}
	$$
	
	$\Leftarrow$ Let $S \subseteq Range(\mathcal{K})$
	
	$$
	\begin{array}{lclll}
		\displaystyle Pr[\mathcal{K}(D') \in S] & =    & \displaystyle \sum_{k \in S} Pr[\mathcal{K}(D') = k] &\quad& \mbox{(by union of elements)} \\
		                                        & \leq & \displaystyle \sum_{k \in S} e^{\delta} Pr[\mathcal{K}(D'') = k] &\quad& \mbox{(by Definition~\ref{def:diff-privacy-2})} \\
		                                        & =    & \displaystyle e^{\delta} \sum_{k \in S} Pr[\mathcal{K}(D'') = k] &\quad& \mbox{(by distributivity)} \\
		                                        & =    & \displaystyle e^{\delta} Pr[\mathcal{K}(D'') \in S] &\quad& \mbox{(by union of elements)} \\
	\end{array}
	$$
		
\end{proof}

\begin{lemma}[Lemma ~\ref{differ} in the paper]
	Let $\mathcal{K}$ be a function that gives $\delta$-differential privacy for all pairs of adjacent databases. Given two databases $D'$ and $D''$ that have the same number of entries and differ in the value of at most one entry, then:
	$$ Pr[\mathcal{K}(D') = k] \leq e^{2\delta} \times Pr[\mathcal{K}(D'') = k] $$
\end{lemma}

\begin{proof}
	Let us call $D$ the common part that $D'$ and $D''$ share, and let us call $r'$ and $r''$ the entries in which they differ, in such a way that $D'= D \cup \{r'\}$ and 
	$$
\renewcommand{\arraystretch}{1.5}
	\begin{array}{lcll}
		\displaystyle Pr[\mathcal{K}(D \cup \{r'\}) = k] & \leq &  \displaystyle e^{\delta} \times Pr[\mathcal{K}(D) = k] & \quad\mbox{(by Definition~\ref{def:diff-privacy-2})} \\
		                                                 & \leq &  \displaystyle e^{\delta} \times e^{\delta} \times Pr[\mathcal{K}(D \cup \{r''\}) = k] & \quad\mbox{(by Definition~\ref{def:diff-privacy-2})} \\
		                                                 & \leq &  \displaystyle e^{2\delta} \times Pr[\mathcal{K}(D'') = k] & \mbox{} \\
	\end{array}
	$$
\end{proof}

\begin{theorem}[Theorem \ref{theo:Shannon} in the paper]
	If a randomized function $\mathcal{K}$ gives $\delta$-differential privacy according to Definition~\ref{def:diff-privacy-2}, then for every result $x^*$ of the function $f$ the Shannon mutual information between the true answers $X$ (i.e. the results of $f$) and the reported answers $Y$  (i.e. the results of $\mathcal{K}$) is bounded by:
	$$I(X;Y) \leq (e^{2\delta} + e^{-2\delta})\delta \log(e) + (e^{2\delta} - e^{-2\delta})\sum_{y}p(y|x^*) \log(p(y|x^*))$$
\end{theorem}

\begin{proof}
	Let us calculate the Shannon mutual information using the formula $I(X;Y) = H(Y) - X(Y|X)$.
	
	\begin{equation}		
		\label{eq:shannon-entropy}
		\begin{array}{lcll}
			\displaystyle H(Y) & =    & \displaystyle -\sum_{y}{p(y)\log p(y)} & \mbox{(by definition)} \\
			     							 & =    & \displaystyle -\sum_{y} \left( \sum_{x}p(x,y) \right) \log \left( \sum_{x}p(x,y) \right) & \mbox{(by probability laws)} \\
			                   & =    & \displaystyle -\sum_{y} \left( \sum_{x}p(x)p(y|x) \right) \log \left( \sum_{x}p(x)p(y|x) \right) & \mbox{(by probability laws)} \\
			                   & \leq & \displaystyle -\sum_{y} \left( \sum_{x}p(x)e^{-2\delta}p(y|x^*) \right) \log \left( \sum_{x}p(x)e^{-2\delta}p(y|x^*) \right) & \mbox{(by Definition~\ref{def:diff-privacy-2} and Lemma~\ref{differ})} \\
			                   & =    & \displaystyle -\sum_{y} e^{-2\delta}p(y|x^*) \left( \sum_{x}p(x) \right) \log \left( e^{-2\delta}p(y|x^*) \sum_{x}p(x) \right) & \mbox{} \\
			                   & =    & \displaystyle -\sum_{y} e^{-2\delta}p(y|x^*) \log (e^{-2\delta}p(y|x^*)) & \mbox{(by probability laws)} \\
			                   & =    & \displaystyle -\sum_{y} \left( e^{-2\delta}p(y|x^*)\log e^{-2\delta} \right) -\sum_{y} \left( e^{-2\delta} p(y|x^*) \log p(y|x^*)  \right) & \mbox{(by distributivity)} \\
	              		     & =    & \displaystyle  -e^{-2\delta}\log e^{-2\delta} \left( \sum_{y} p(y|x^*) \right) -\sum_{y} \left( e^{-2\delta} p(y|x^*) \log p(y|x^*)  \right) & \mbox{} \\
			                   & =    & \displaystyle  \delta e^{-2\delta} \log{e} -e^{-2\delta} \sum_{y} p(y|x^*) \log p(y|x^*) & \mbox{(by probability laws)} \\
		\end{array}
	\end{equation}

	\begin{equation}		
		\label{eq:shannon-cond-entropy}
		\begin{array}{lcll}
			\displaystyle H(Y|X) & =    & \displaystyle -\sum_{x} p(x) \sum_{y} p(y|x) \log p(y|x) & \mbox{(by definition)} \\	       
			                     & \geq & \displaystyle -\sum_{x} p(x) \sum_{y} e^{2\delta} p(y|x^*) \log(e^{2\delta} p(y|x^*)) & \mbox{(by Definition~\ref{def:diff-privacy-2} and Lemma ~\ref{differ})} \\	       
			                     & =    & \displaystyle -\left( \sum_{y} e^{2\delta} p(y|x^*) \log(e^{2\delta}p(y|x^*))  \right) \sum_{x}p(x)  & \mbox{(by distributivity)} \\	       
			                     & =    & \displaystyle - \sum_{y} e^{2\delta} p(y|x^*) \log(e^{2\delta}p(y|x^*)) & \mbox{(by probability laws)} \\	 
			                     & =    & \displaystyle - \sum_{y} \left( e^{2\delta} p(y|x^*) \log(e^{2\delta}) \right) -\sum_{y} \left( e^{2\delta} p(y|x^*) \log p(y|x^*) \right) & \mbox{} \\	 
			                     & =    & \displaystyle -e^{\delta} \log e^{2\delta} \left( \sum_{y}p(y|x^*) \right) -e^{2\delta}\sum_{y}p(y|x^*) \log p(y|x^*) & \mbox{} \\	 
			                     & =    & \displaystyle -\delta e^{2\delta} \log e - e^{2\delta}\sum_{y}p(y|x^*)\log p(y|x^*) & \mbox{(by probability laws)} \\	
		\end{array}
	\end{equation}
	
	\[	\begin{array}{lcll}
			\displaystyle I(X;Y) & =    & \displaystyle H(Y) - X(Y|X) & \mbox{(by definition)} \\
			                     & \leq & \displaystyle \delta e^{-2\delta} \log{e} -e^{-2\delta} \sum_{x} p(y|x^*) \log p(y|x^*) + \\
			                     &      & \displaystyle 2\delta e^{2\delta} \log e + e^{2\delta}\sum_{y}p(y|x^*)\log p(y|x^*) & \mbox{(by Equations~\ref{eq:shannon-entropy} and ~\ref{eq:shannon-cond-entropy})} \\
			                     & =    & \displaystyle (e^{2\delta} + e^{-2\delta})\delta \log(e) + (e^{2\delta} - e^{-2\delta})\sum_{y}p(y|x^*) \log(p(y|x^*)) & \mbox{(by distributivity)} \\
		\end{array}
		\]
	
\end{proof}

\begin{theorem}[Theorem~\ref{theo:Renyi} in the paper]
	If a randomized function $\mathcal{K}$ gives $\delta$-differential privacy according to Definition~\ref{def:diff-privacy-2}, then the R\'enyi min mutual information between the true answer of the function $X$ and the reported answer $Y$ is bounded by:
	$$I_{\infty}(X;Y) \leq 2\delta \log e.$$
\end{theorem}

\begin{proof}

	Let us calculate the R\'enyi mutual information using the formula $I_\infty(X;Y) = H_\infty(X) - X_\infty(X|Y)$.

	\begin{equation}		
		\label{eq:renyi-entropy}
		\begin{array}{lcll}
			\displaystyle H_{\infty}(X) & = & \displaystyle -\log \max_{x} p(x) &\quad \mbox{(by definition)} \\
		\end{array}
	\end{equation}	
	\begin{equation}		
		\label{eq:renyi-cond-entropy}
		\begin{array}{lcll}
			\displaystyle H_{\infty}(X|Y) & =    & \displaystyle -\log \sum_{y} p(y) \max_{x} p(x|y) &\quad \mbox{(by definition)} \\
			                              & =    & \displaystyle -\log \sum_{y} \max_{x} p(y)p(x|y) &\quad \mbox{} \\
			                              & =    & \displaystyle -\log \sum_{y} \max_{x} p(x)p(y|x) &\quad \mbox{(by probability laws)} \\
			                              & \geq & \displaystyle -\log \sum_{y} \max_{x} p(x)e^{2\delta}p(y|x^*) &\quad \mbox{(by Definition~\ref{def:diff-privacy-2} and Lemma~\ref{differ})} \\
			                              & =    & \displaystyle -\log \sum_{y} e^{2\delta} p(y|x^*) \max_{x} p(x) &\quad \mbox{} \\
			                              & =    & \displaystyle -\log \left( e^{2\delta} \max_{x} p(x) \sum_{y} p(y|x^*)  \right) &\quad \mbox{} \\
			                              & =    & \displaystyle -\log \left( e^{2\delta} \max_{x} p(x) \right) &\quad \mbox{(by probability laws)} \\
			                              & =    & \displaystyle -2\delta \log e - \log \max_{x} p(x)   & \quad\mbox{} \\
		\end{array}
	\end{equation}
\renewcommand{\arraystretch}{1.5}	
	\[
		\begin{array}{lcll}
			\displaystyle I_{\infty}(X;Y) & =    & \displaystyle H_{\infty}(X) - H_{\infty}(X|Y) & \mbox{(by definition)} \\
			                              & \leq & \displaystyle - \log \max_{x} p(x) + 2 \delta \log e + \log \max_{x} p(x) & \mbox{(by Equations~\ref{eq:renyi-entropy} and ~\ref{eq:renyi-cond-entropy})} \\
			                              & =    & \displaystyle 2\delta \log e & \mbox{} \\	
		\end{array}
	\]
\end{proof}

\end{document}